\documentclass[11pt]{article}
\usepackage[letterpaper,margin=1in]{geometry}
\usepackage[T1]{fontenc}
\usepackage[utf8]{inputenc}
\usepackage[backend=biber, style=alphabetic, maxbibnames=99]{biblatex}
\usepackage{lmodern,microtype,amsmath,amssymb,amsthm,complexity}
\usepackage{booktabs,enumitem,needspace}
\usepackage{xcolor}
\usepackage{hyperref}
\numberwithin{equation}{section}
\usepackage[nameinlink]{cleveref}
\usepackage{authblk}

\definecolor{navy}{HTML}{000099}
\definecolor{teal}{HTML}{0F766E}
\hypersetup{
    colorlinks=true,
    citecolor=teal,
    linkcolor=magenta,
}

\allowdisplaybreaks[1]

\newtheorem{theorem}{Theorem}[section]
\newtheorem{lemma}[theorem]{Lemma}
\newtheorem{proposition}[theorem]{Proposition}
\newtheorem{corollary}[theorem]{Corollary}
\newtheorem{claim}[theorem]{Claim}
\newtheorem*{mainresult}{Theorem~\ref{thm:general}}
\newtheorem*{mainbound}{Theorem~\ref{thm:main}}
\newtheorem*{sparsestate}{Lemma~\ref{lem:sparse-state}}
\crefname{claim}{claim}{claims}
\Crefname{claim}{Claim}{Claims}
\theoremstyle{definition}
\newtheorem{definition}[theorem]{Definition}
\theoremstyle{remark}
\newtheorem{remark}[theorem]{Remark}

\newcommand{\Tr}{\operatorname{Tr}}
\newcommand{\trans}{\mathsf{T}}
\renewcommand{\Re}{\operatorname{Re}}

\newcommand{\conv}{\operatorname{conv}}
\newcommand{\Var}{\operatorname{Var}}
\newcommand{\op}{\mathrm{op}}
\newcommand{\ket}[1]{|#1\rangle}
\newcommand{\bra}[1]{\langle #1|}

\newcommand{\pairs}{E}
\newcommand{\vx}{x}
\newcommand{\vy}{y}

\title{\bf Improved classical optimization of fermionic\\ Hamiltonians beyond Gaussian states}

\author{Ojas Parekh\thanks{\href{mailto:odparek@sandia.gov}{odparek@sandia.gov}}}
\affil{Quantum Algorithms and Applications Collaboratory,\\Sandia National Laboratories,\\Albuquerque, NM, USA}

\date{}

\begin{document}
\maketitle
\vspace{-1.5em}
\begin{abstract}
We give a polynomial-time randomized classical algorithm that approximates the largest eigenvalue $\lambda_{\max}$ of a traceless quartic fermionic Hamiltonian $H$ on $n$ modes. The algorithm achieves energy at least $\lambda_{\max}/O(\sqrt n)$ and outputs either a fermionic Gaussian state or a polynomially sparse state in the Fock basis. This is achieved through a new upper bound on $\lambda_{\max}$ in terms of the maximum energy of a Gaussian state and $\sqrt{2^{-n} \Tr[H^2]}$.
\end{abstract}

\section{Introduction}\label{sec:statement}
Finding the extremal eigenvalues of an $n$-mode quartic fermionic Hamiltonian is a basic
optimization problem in quantum information and quantum many-body physics~\cite{McArdle2020,Arovas2022}.
Even when the Hamiltonian is a polynomial of low degree, it acts on a
space of dimension exponentially large in $n$. Computing extremal eigenvalues of such Hamiltonians to inverse-polynomial additive precision in $n$ is \QMA-hard~\cite{LCV2007}, so we turn to coarser approximations in order to obtain efficient algorithms.

An approximation algorithm is a polynomial-time classical or quantum algorithm that approximates the maximum\footnote{This is equivalent to finding the minimum eigenvalue of $-H$ for traceless Hamiltonians $H$, as we consider here.} eigenvalue within a provable multiplicative approximation ratio. Approximation algorithms typically focus on two questions: how much of the optimal energy can be
provably achieved, and how to efficiently describe and compute properties of a state achieving that energy. Fermionic Gaussian states offer a
natural answer to the second question: they admit a compact description by a
matrix of quadratic moments, from which higher-order moments
can be evaluated by Wick's theorem~\cite{Bravyi}.

For traceless fermionic Hamiltonians with quadratic and quartic terms, Bravyi, Gosset, K\"onig, and Temme~\cite{BGKT} give a randomized classical approximation algorithm using Gaussian states to achieve an $\Omega(1/(n\log n))$ approximation ratio. A natural question is whether constant-factor approximations are possible using Gaussian states, which is not the case. Hastings and O'Donnell~\cite{HO} show that the optimal approximation ratio using Gaussian states is $\Theta(1/\sqrt{n})$ on instances of the quartic Sachdev--Ye--Kitaev (SYK) model. This is a well-studied random fermionic model~\cite{DKKA,BassoChenDalzell2024}. Hastings and O'Donnell~\cite{HO} also give a quantum $\Omega(1)$-approximation for the SYK model; however, their $\Omega(1/\sqrt{n})$-approximation is the best-known classical approximation algorithm. 

This leaves a substantial gap for approximating general quartic fermionic Hamiltonians. We improve the general guarantee to
$\Omega(1/\sqrt{n})$ by supplementing Gaussian states with a simple non-Gaussian family.

\subsection{Problem and main result}
A system of $n$ fermionic modes has $d=2n$ Majorana operators
$c_1,\ldots,c_d$ acting on a Hilbert space $\mathcal H$ of dimension $2^n$.
They satisfy the canonical anticommutation relations (CAR):
\begin{equation}
 c_i^\dagger=c_i,\quad c_ic_j+c_jc_i=2\delta_{ij}I_{\mathcal H}.
 \label{eq:car}
\end{equation}
Let $[d]=\{1,\ldots,d\}$. For $S=\{s_1,\ldots,s_k\} \subseteq [d]$ with $s_1<\cdots<s_k$, let
$c_S=c_{s_1}\cdots c_{s_k}$, with $c_\varnothing=I_{\mathcal H}$.
We consider the homogeneous quartic Hamiltonian
\begin{equation}
 H=\sum_{S\subseteq[d]:|S|=4}w_Sc_S.
 \label{eq:input}
\end{equation}
The coefficients are rational, and the ordered quartic Majorana monomials are
Hermitian and traceless. Let $t=|\{S\subseteq[d]:|S|=4,\,w_S\ne0\}|$ be the number of nonzero coefficients.
For a density matrix $\rho$, let
$\langle A\rangle_\rho=\Tr_{\mathcal H}[\rho A]$, and use
$\langle A\rangle_\psi=\bra\psi A\ket\psi$ for a state $\ket{\psi}$.
Our objective is $\lambda_{\max}(H)=\max_\rho\langle H\rangle_\rho$.
As in approximation algorithms for classical problems, an $\alpha$-approximation produces a solution of value at least $\alpha$ times the optimum in polynomial time. For our problem, this means producing a classical description of a quantum state whose energy is classically computable in polynomial time.

\begin{theorem}\label{thm:general}
There is a randomized polynomial-time classical algorithm that, given
a Hamiltonian $H$ of the form \eqref{eq:input} on $n\ge2$ modes, returns a description of a state $\ket{\psi}$ satisfying
\begin{equation}
 \langle H\rangle_\psi\ge
 \frac{\lambda_{\max}(H)}{c_G\log(2n)+80\sqrt n}
 \label{eq:general-ratio}
\end{equation}
with probability at least $2/3$, for a universal constant $c_G$.
The output is either a fermionic Gaussian state or a superposition
of at most $t+1$ Fock basis states.
Its energy is efficiently computable classically.
The running time is polynomial in $n$ and the input length.
\end{theorem}

A Fock basis state is specified by an $n$-bit string specifying occupied modes. We describe the superposition in
\Cref{thm:general} by listing these strings along with their amplitudes.

We can also allow quadratic terms in $H$ using the reduction of Bravyi, Gosset,
K\"onig, and Temme~\cite[Lemma~6 and Appendix~C]{BGKT}. It converts a
traceless Hamiltonian with quadratic and quartic terms into a homogeneous quartic
Hamiltonian with one additional mode, preserving the maximum energy and
number of nonzero terms. Our guarantee therefore holds with $n$ replaced
by $n+1$. The reduction recovers a state on the original modes without
energy loss and preserves Gaussianity. 

For $H$ in \eqref{eq:input}, define the normalized trace and Frobenius norm of $H$ and its Gaussian optimum by
\begin{equation}
 \begin{aligned}
 \tau[A]&:=2^{-n}\Tr_{\mathcal H}[A],\\
 \sigma&:=2^{-n/2}\|H\|_F=\sqrt{\tau[H^2]}=\sqrt{\textstyle\sum_{S}w_S^2},\\
 G(H)&:=\max_{\ket{g}\text{ Gaussian}}\langle H\rangle_g.
 \end{aligned}
 \label{eq:scales}
\end{equation}
We also use $G(X)$ for the maximum Gaussian expectation of any Hermitian
operator $X$. \Cref{sec:majorana-monomials} verifies the trace identity regarding $\sigma$ in
\eqref{eq:scales}. The family of Gaussian states includes the Fock basis, so
$G(H)\ge\tau[H]=0$.

The approximation guarantee in \Cref{thm:general} stems from a bound on $\lambda_{\max}(H)$ in terms of the Gaussian optimum and the normalized
Frobenius norm.
\begin{theorem}\label{thm:main}
For every homogeneous quartic Hamiltonian $H$ on $n\ge2$ modes,
\begin{equation}
 \lambda_{\max}(H)\le G(H)+5\sqrt n\,\sigma.
 \label{eq:main-additive}
\end{equation}
\end{theorem}
This bound lets us infer that if $G(H)$ is large relative to $\lambda_{\max}(H)$, then the best Gaussian state is a good approximation; otherwise $\sqrt{n}\sigma$ is large, and an approximation with energy $\Omega(\sigma)$ gives an
$\Omega(1/\sqrt{n})$ approximation ratio. Thus our algorithm produces two candidate states to address each of these cases.   

The first candidate state is a Gaussian state delivering an $\Omega(1/\log n)$-approximation of the optimal Gaussian energy, produced by a classical approximation algorithm of Bravyi, Gosset, K\"onig, and Temme~\cite[Theorem~4]{BGKT}. Their algorithm is based on an approximation algorithm for quadratic optimization with orthogonality constraints by So~\cite{So}. A simple classical algorithm produces the desired second candidate state.
\begin{lemma}\label{lem:sparse-state}
There is a randomized polynomial-time classical algorithm that, given
a Hamiltonian $H$ of the form \eqref{eq:input} on $n\ge2$ modes, returns,
with probability at least $5/6$, a description of a state $\ket{\psi_F}$
satisfying
\begin{equation}
 \langle H\rangle_{\psi_F}\ge\sigma/16.
 \label{eq:baseline}
\end{equation}
The state is supported on at most
\begin{equation*}
 \min\left\{t+1,\ 1+\binom n2+\binom n4\right\}
\end{equation*}
Fock basis strings. For $H\ne0$, it is proportional to
$(16\sigma I+H)\ket{\vy}$ for some $\vy\in\{0,1\}^n$.
Its energy is efficiently computable classically.
The running time is polynomial in $n$ and the input length.
\end{lemma}

\Cref{sec:coefficient} proves \Cref{thm:main}, \Cref{sec:filter} proves \Cref{lem:sparse-state}, and \Cref{sec:algorithm} proves \Cref{thm:general}. The algorithm of \Cref{thm:general} is as follows.

\paragraph{Algorithm.}
\begin{enumerate}[leftmargin=*,itemsep=2pt,topsep=3pt]
\item Run the algorithm of \Cref{lem:sparse-state} to obtain the
sparse Fock-basis candidate $\ket{\psi_F}$.
\item Run the Bravyi--Gosset--K\"onig--Temme Gaussian-state approximation algorithm twice independently and keep the Gaussian candidate with the higher energy.
\item Return the candidate state with higher energy, or the Gaussian candidate if the first step fails.
\end{enumerate}

\begin{remark}
The numerical constants in \eqref{eq:main-additive} and
\eqref{eq:general-ratio} can be improved without changing the algorithm.
We retain the constants above for the sake of simpler proofs.
\end{remark}

\subsection{Related work}
Beyond the related results discussed above, Herasymenko, Stroeks, Helsen, and Terhal~\cite{HSHT} give efficient
constant-factor Gaussian approximations for fermionic Hamiltonians with quadratic and quartic terms that are sparse, namely each Majorana occurs in a bounded number $k$ of terms. They also give an efficient constant-factor Gaussian approximation for sparse quartic SYK with bounded expected degree. Hothem, Parekh, and Thompson~\cite{HPT} improve
the dependence on $k$ in the approximation ratio. For fixed even $q$, these results generalize to sparse homogeneous degree-$q$ Hamiltonians. 

Stroeks, Terhal, and Herasymenko~\cite{STH}
show that Gaussian states achieve at least one third of the optimum
for traceless Hamiltonians with repulsive density--density interactions,
real number-conserving hopping, and on-site potentials. This guarantee is existential; they give efficient algorithms under additional assumptions, including a bipartite density--density interaction graph.

Ding, King, Kiani, and Anschuetz~\cite{DKKA} extend dense-SYK results~\cite{HTS,HO} to sparse random instances.
For quartic SYK with interaction retention probability
$p=\Omega(\log n/n^2)$, the Gaussian approximation ratio is
$\Theta(1/\sqrt{n})$ with high probability. They also show that the quantum
algorithm of Hastings and O'Donnell achieves a constant ratio under
the stronger condition $p=\Omega(\log n/n)$.

\subsection{Technical overview}

A real quartic moment matrix $M^\rho$ of a state $\rho$ is a symmetric matrix whose rows and columns are indexed by pairs of indices $i < j$, corresponding to quadratic Hermitian terms $ic_ic_j$, with
\begin{equation*}
M^\rho_{(i,j),(k,\ell)} = -\Re \langle c_ic_jc_kc_\ell\rangle_{\rho}.
\end{equation*}
Such moment matrices are central objects in non-commutative moment-sum-of-squares hierarchies. See \Cref{sec:pairs} for more details. A real symmetric $A$ of the same size as $M^\rho$ is called a coefficient matrix, and its entries represent coefficients of a quartic Hamiltonian $H_A$. Moreover,
\begin{equation}\label{eq:expectation-trace-OV}
\langle H_A\rangle_\rho=\Tr[AM^\rho],\quad \tau[H_A]=\Tr[A].   
\end{equation}

One of our main technical results shows that the moment matrix $M^\rho$ can be approximated with that of a mixture $\omega$ of Gaussian states:
\begin{equation}\label{eq:moment-bound-OV}
\|M^\rho-M^\omega\|_\op\le2,
\end{equation}
proven in \Cref{sec:ensemble}. 

\paragraph{PSD coefficient matrices.} Our bound \eqref{eq:moment-bound-OV} implies $M^\rho \preceq M^\omega + 2I$, and in particular 
\begin{equation}\label{eq:moment-coefficient-bound-OV}
\Tr[AM^\rho] \leq \Tr[AM^\omega] + 2\Tr[A],
\end{equation}
for positive semidefinite coefficient matrices $A$. In this case $H_A$ is a sum of squares of quadratic Majorana Hamiltonians. The bound \eqref{eq:moment-coefficient-bound-OV} along with \eqref{eq:expectation-trace-OV} gives 
\begin{equation}\label{eq:expectation-bound-OV}
 \langle H_A \rangle_\rho \leq \langle H_A \rangle_\omega + 2\tau[H_A].
\end{equation}
If $\omega^*$ is the Gaussian state maximizing $\langle H_A \rangle_\omega$, then $\langle H_A \rangle_{\omega^*} \geq \langle H_A \rangle_\omega$ and also
\begin{equation}\label{eq:maximally-mixed-bound-OV}
\langle H_A \rangle_{\omega^*} \geq \langle H_A \rangle_{2^{-n}I} = \tau[H_A],
\end{equation}
since the maximally mixed state, $2^{-n}I$ is a mixture of Gaussian states. \Cref{eq:expectation-bound-OV,eq:maximally-mixed-bound-OV} then give
\begin{equation*}
\langle H_A \rangle_{\omega^*} \geq \tfrac{1}{3}\langle H_A \rangle_\omega + \tfrac{2}{3}\tau[H_A] \geq \tfrac{1}{3}\langle H_A \rangle_\rho,
\end{equation*}
showing that there exists a Gaussian state earning $1/3$ the maximum energy. We improve this to $1/2$ using a stronger bound in \Cref{cor:psd}. 

\paragraph{A Frobenius-norm bound.}

For general non-PSD $A$, we use \eqref{eq:moment-bound-OV} to prove the
eigenvalue bound 
\begin{equation}\label{eq:eigenvalue-bound-OV}
 \lambda_{\max}(H)\le G(H)+5\sqrt n\,\sigma.
\end{equation}
of \Cref{thm:main}. The key step is converting
\eqref{eq:moment-bound-OV} into a Frobenius-norm bound:
\begin{equation}\label{eq:frobenius-bound-OV}
\|M^\rho-M^\omega\|_\op\le2 \implies \|M^\rho-M^\omega\|_F\le5\sqrt{6n},
\end{equation}
which is proven in \Cref{sec:coefficient}. We improve upon more generic relationships between these two norms by leveraging special properties of moment matrices arising from physical states. The $\Theta(\sqrt{n})$ factor in \eqref{eq:frobenius-bound-OV} carries through to \eqref{eq:eigenvalue-bound-OV} and is ultimately responsible for the approximation guarantee.

\Cref{sec:filter} proves \Cref{lem:sparse-state} by applying
$16\sigma I+H$ to a uniformly random Fock basis state, repeating until the normalized output has energy at least $\sigma/16$. Polynomially many independent trials suffice with probability at least 5/6. With \eqref{eq:eigenvalue-bound-OV} in hand, \Cref{sec:algorithm} combines this construction with the Gaussian-state approximation algorithm of Bravyi, Gosset, K\"onig, and Temme~\cite{BGKT} to prove \Cref{thm:general}.

\paragraph{Gaussian variance bounds.} Our main technical lemma, \Cref{lem:local-overlap}, approximates a state $\ket{\psi}$ by choosing a Gaussian $\ket{g}$ that maximizes the overlap $|\langle g|\psi\rangle|$, with its phase chosen so that $a := \langle g|\psi\rangle > 0$. For a quadratic Majorana Hamiltonian $Q$, the relevant approximation error is
\begin{equation*}
    \left| \frac{\langle g|Q^2|\psi\rangle}{a}-\langle Q^2\rangle_g \right|
 \le \Var_g(Q),
\end{equation*}
where the latter is the variance of $\ket{g}$ with respect to $Q$. The relevance to approximating eigenvalues is more apparent by choosing $\ket{\psi}$ to be a top eigenvector of $Q^2$, corresponding to its maximum eigenvalue. In this case we get
\begin{equation*}
 \lambda_{\max}(Q^2) = \frac{\langle g|Q^2|\psi\rangle}{a},
\end{equation*}
so that
\begin{equation*}
 \lambda_{\max}(Q^2)-G(Q^2)
 \le\frac{\langle g|Q^2|\psi\rangle}{a}-\langle Q^2\rangle_g
 \le \Var_g(Q).
\end{equation*}
This can be extended to a homogeneous quartic Hamiltonian $H$ by expressing it as a signed sum of squares of quadratic Hamiltonians, which is always possible and covered in \Cref{lem:scalar-comparison}. The main issue is that $\ket{g}$ now depends on $H$. In contrast, the Gaussian mixture $\omega$ in our desired bound \eqref{eq:moment-bound-OV} only depends on an arbitrary state $\rho$. We remove this dependence on $H$ by taking a suitable mixture of Gaussians in \Cref{prop:ensemble}.

\section{Preliminaries}\label{sec:preliminaries}
All states specified in ket notation have unit norm unless otherwise specified. We use the term vector when normalization is not assumed.
Adjoints are denoted by $\dagger$, transposes by $\trans$, and complex
conjugation by an overline. For a matrix $A$, let $A_i$ denote its $i$th column. Vector norms are Euclidean by default, and $\|\cdot\|_\op$ and
$\|\cdot\|_F$ denote operator and Frobenius norms. We say $A\preceq B$
when $B-A$ is positive semidefinite (PSD). The trace $\Tr$ is always
unnormalized, and $\tau=2^{-n}\Tr_{\mathcal H}$ is the normalized trace on the underlying Hilbert space $\mathcal{H}$. We use $I$ for the identity, with $I_{\mathcal H}$ for the identity operator on $\mathcal H$ and $I_k$ for the $k\times k$ identity matrix when the space or dimension needs to be specified. Unless stated otherwise, a Hamiltonian $H$ is homogeneous quartic, and the number of modes is $n\ge 2$.

As in Hastings--O'Donnell~\cite[Section~2.7]{HO}, we carry out numerical
optimization and energy evaluation to sufficiently small inverse-polynomial
additive accuracy after coefficient normalization. The approximation
guarantees have slack to absorb this numerical error.

\subsection{Majorana monomials}\label{sec:majorana-monomials}
The CAR imply
\begin{equation*}
 c_S^\dagger=(-1)^{|S|(|S|-1)/2}c_S,
 \quad c_S^2=(-1)^{|S|(|S|-1)/2}I.
\end{equation*}
In particular, quartic monomials are Hermitian involutions. The ordered monomials
form an orthonormal operator basis for the normalized trace inner product:
\begin{equation}
 \tau[c_S^\dagger c_T]=\delta_{S,T}.
 \label{eq:orthogonal-words}
\end{equation}
Indeed, every nonempty monomial anticommutes with some Majorana and hence has
zero trace, while $c_S^\dagger c_T$ is a signed monomial on
$S\mathbin\triangle T$. We use the notation $[X]_k$ for the degree-$k$ part of the Majorana
expansion of $X$. Since quartic monomials satisfy
$c_Sc_T=(-1)^{|S\cap T|}c_Tc_S$, the square of a homogeneous quartic $R=\sum_{S:|S|=4}r_Sc_S$ with real coefficients $r_S$
has only degrees zero, four, and eight, with $[R^2]_0=\tau[R^2]I$. Trace orthonormality~\eqref{eq:orthogonal-words} implies $\tau[R^2]=\sum_{|S|=4}r_S^2$, which is also used in \eqref{eq:scales}.

\subsection{Moment and coefficient matrices}\label{sec:pairs}

Moment matrices arise naturally in non-commutative moment--sum-of-squares
hierarchies of semidefinite programs (SDPs), where they are dual to coefficient matrices through the trace
inner product (see \eqref{eq:moment-SoS-duality} below). For background see
Pironio, Navascu\'es, and Ac\'in~\cite[Sections~3.1 and~3.4]{PNA2010}. For an account in the setting of local Hamiltonians, see
Parekh and Thompson~\cite[Section~2.2 and Appendix~A]{PT2022}. Our analysis works directly with moment matrices and shows that the quartic moment matrix of a state can be approximated by that of a Gaussian mixture.

\begin{definition}[Quadratic moments and quartic moment matrix]
\label{def:moments}
Let $\pairs=\{(i,j):1\le i<j\le d\}$, $m=\binom d2$, and
$q_{ij}=ic_ic_j$. Set $\boldsymbol q=(q_p)_{p\in\pairs}$. For $\vx\in\mathbb R^{\pairs}$ and a real symmetric
$m\times m$ coefficient matrix $A$, define
\begin{equation}
 Q_{\vx}=\sum_{p\in\pairs}x_pq_p,
 \quad H_A=\boldsymbol q^{\trans}A\boldsymbol q
            :=\sum_{p,r\in\pairs}A_{pr}q_pq_r.
 \label{eq:SA}
\end{equation}
The operators $q_p$, $Q_{\vx}$, and $H_A$ act on $\mathcal H$ and are Hermitian. For any Hermitian operator $X$ on $\mathcal H$,
define its real vector of quadratic moments and real symmetric
\emph{quartic moment matrix} by
\begin{equation*}
 \mu^X=(\,\Tr[Xq_p]\,)_{p\in E},
 \quad M^X=(\,\Re\Tr[Xq_pq_r]\,)_{p,r\in E},
\end{equation*}
respectively. We will frequently take $X$ to be a state $\rho$. For a
state $\ket{\psi}$, we abbreviate $\mu^{\ket\psi\bra\psi}$ and
$M^{\ket\psi\bra\psi}$ by $\mu^\psi$ and $M^\psi$, respectively.
\end{definition}

For $p=(i,j)$ and $r=(k,\ell)$, the CAR give
$q_pq_r=-c_ic_jc_kc_\ell$ and
$q_pq_r=(-1)^{|\{i,j\}\cap\{k,\ell\}|}q_rq_p$.
Thus $q_pq_r$ is Hermitian when the overlap between $p$ and $r$ is even
and skew-Hermitian when it is odd. Its trace against the Hermitian
operator $X$ is correspondingly real or purely imaginary, so
\begin{equation}
 M^X_{(i,j),(k,\ell)}=-\Re\Tr[Xc_ic_jc_kc_\ell]=
 \begin{cases}
 -\Tr[Xc_ic_jc_kc_\ell]& |\{i,j\}\cap\{k,\ell\}|=0,\\
 0& |\{i,j\}\cap\{k,\ell\}|=1,\\
  \Tr[X]& (i,j)=(k,\ell).\\
 \end{cases}
 \label{eq:moment-entry}
\end{equation}
Thus $M^X$ records identity
and quartic moments, and $\mu^X$ records quadratic moments.
Each quartic moment appears three times in the upper triangle of $M^X$
and six times in the full matrix, up to the sign needed to order its four
indices. In particular, when $i<j<k<\ell$ we have
$M^X_{(i,j),(k,\ell)}=-\Tr[Xc_ic_jc_kc_\ell]$ by \eqref{eq:moment-entry}.

For a density matrix $\rho$, $M^\rho$ has unit diagonal entries. \Cref{def:moments} implies the identity
\begin{equation}
 \langle H_A\rangle_\rho=\Tr[AM^\rho].
 \label{eq:moment-SoS-duality}
\end{equation} 
Moreover, $\tau[H_A]=\Tr[A]$ by taking $\rho = 2^{-n}I$.

Since $Q_{\vx}^2 = H_{\vx\vx^{\trans}}$, this gives
\begin{equation*}
 \langle Q_{\vx}^2\rangle_\rho
 =\Tr[\vx\vx^{\trans}M^\rho]
 =\vx^{\trans}M^\rho\vx,
 \quad \langle Q_{\vx}\rangle_\rho=\vx^{\trans}\mu^\rho.
\end{equation*}
Consequently,
\begin{equation}
 \vx^{\trans}\bigl(M^\rho-\mu^\rho(\mu^\rho)^{\trans}\bigr)\vx
 =\langle Q_{\vx}^2\rangle_\rho-\langle Q_{\vx}\rangle_\rho^2
 =\bigl\langle(Q_{\vx}-\langle Q_{\vx}\rangle_\rho I)^2\bigr\rangle_\rho
 =\Var_\rho(Q_{\vx})\ge0.
 \label{eq:moment-psd}
\end{equation}
If $A=\sum_r\vx_r\vx_r^{\trans}\succeq0$,
then $H_A=\sum_r Q_{\vx_r}^2\succeq0$ is a sum of squares of quadratic Hamiltonians.

An appropriately rescaled moment matrix
also gives a coefficient matrix for homogeneous quartic Majorana polynomials.
\begin{definition}[Coefficient matrix of a homogeneous quartic]\label{def:quartic-coefficient}
Let $R=\sum_{S:|S|=4}r_Sc_S$ have real coefficients, and let
\begin{equation*}
 C^R=M^{2^{-n}R},\quad s^2=\tau[R^2]=\sum_{S:|S|=4}r_S^2.
\end{equation*}
Then $C^R$ is the \emph{coefficient matrix} of $R$ in the sense that $H_{C^R/6} = R$. 
\end{definition}

\begin{claim}\label{clm:quartic-moments}
The entry $C^R_{pr}=0$ whenever $p,r$ share an index, and
\begin{subequations}
\begin{align}
 \Tr[(C^R)^2]&=6s^2, \label{eq:quartic-coefficient-norm}\\
 H_{C^R}&=6R, \label{eq:quartic-reconstruction}\\
 \tau[Rq_pq_r]&=C^R_{pr}. \label{eq:quartic-trace-pairing}
\end{align}
\end{subequations}
\end{claim}
\begin{proof}
Take $C=C^R$. For disjoint $p=(i,j)$ and $r=(k,\ell)$, let $S=\{i,j,k,\ell\}$ and
$c_ic_jc_kc_\ell=\varepsilon c_S$, where $\varepsilon\in\{\pm1\}$.
By \eqref{eq:orthogonal-words} and \eqref{eq:moment-entry},
$C_{pr}=-\varepsilon r_S$, so $C_{pr}q_pq_r=r_Sc_S$.
If $p=r$, then $\tau[Rq_pq_r]=\tau[R]=0$. If the pairs share exactly
one index, their product is proportional to a quadratic Majorana monomial and is
orthogonal to $R$.
This proves that $\tau[Rq_pq_r]$ is real and equals $C_{pr}$ for \eqref{eq:quartic-trace-pairing}. For \eqref{eq:quartic-coefficient-norm} and \eqref{eq:quartic-reconstruction},
each four-element set $S$ has three pairings and their transposes, contributing
$6r_S^2$ to $\Tr[C^2]$ and $6r_Sc_S$ to $H_C$.
\end{proof}

\subsection{Gaussian states}\label{sec:gaussian}
We present the features of fermionic Gaussian states necessary for our results. For a broader background on fermionic operators and the Jordan--Wigner transform, see Nielsen~\cite{Nielsen2005}, and for Gaussian states, see Bravyi~\cite{Bravyi}.

In addition to a Gaussian state, our algorithm constructs a sparse state in a fixed Fock basis: 
\begin{equation*}
 a_j=(c_{2j-1}+ic_{2j})/2,\quad
 \ket{\vy}=(a_1^\dagger)^{y_1}\cdots(a_n^\dagger)^{y_n}\ket{0^n},
 \quad \vy\in\{0,1\}^n,
\end{equation*}
where $\ket{0^n}$ is the common vacuum of the annihilation operators $a_j$.
Thus $ic_{2j-1}c_{2j}=2a_j^\dagger a_j-I$.

\begin{definition}[Pure Gaussian states]\label{def:pure-gaussian}
A state $\ket{g}$ is \emph{Gaussian} if, for some matrix $O$ in
$O(d)$, the real orthogonal group, it is the common vacuum of
$\widetilde a_j=(\widetilde c_{2j-1}+i\widetilde c_{2j})/2$, where
$\widetilde c_a=\sum_iO_{ai}c_i$. We denote this vacuum by $\ket{g}=\ket{\widetilde{0}^{n}}$, so $\widetilde a_j\ket{\widetilde{0}^{n}}=0$ for all $j$. Equivalently, its projector is
\begin{equation}
 \ket{g}\bra{g}
 =\prod_{j=1}^n\frac{I-i\widetilde c_{2j-1}\widetilde c_{2j}}2.
 \label{eq:gaussian-projector-product}
\end{equation}
The cone $\mathcal C_G$ consists of all complex scalar multiples of
Gaussian states, including zero. Its elements need not be normalized.
\end{definition}
The operators $\widetilde c_i$ also satisfy the CAR: they are Hermitian because $O$ is real, and orthogonality gives
\begin{equation*}
 \{\widetilde c_i,\widetilde c_j\}
 =\sum_{k,\ell}O_{ik}O_{j\ell}\{c_k,c_\ell\}
 =2\sum_k O_{ik}O_{jk}I=2\delta_{ij}I.
\end{equation*}
The calculation above applies to every $O\in O(d)$.
The choice $\det O=1$ gives even fermion parity, meaning support only on
Fock basis states with even $\sum_j y_j$; $\det O=-1$ gives odd fermion
parity, with support only on those with odd $\sum_j y_j$.
Allowing both signs includes every Fock basis state. The product formula~\eqref{eq:gaussian-projector-product}~\cite[Section~IV, Equation~(19)]{Bravyi}
expresses the rank-one Gaussian projectors $\ket{g}\bra{g}$ as a
continuous image of the compact group $O(d)$, hence a compact set. This ensures attainment of the maximum overlap used in
\Cref{lem:local-overlap,lem:scalar-comparison}. It also makes the convex hull of Gaussian moment matrices compact, as needed for the strict-separation argument in \Cref{prop:ensemble}. A convex mixture of these projectors need not itself be Gaussian.

\begin{definition}[Gaussian covariance matrix]
\label{def:gaussian-covariance}
The covariance matrix of a Gaussian $\ket{g}$ is the real skew $d\times d$ matrix
\begin{equation*}
 \Gamma^g = (\ \tfrac i2\langle c_ic_j-c_jc_i\rangle_g\ )_{i,j \in [d]}.
\end{equation*}
Thus 
\begin{equation}
 \langle c_ic_j\rangle_g=\delta_{ij}-i\Gamma^g_{ij},
 \label{eq:two-point}
\end{equation}
and the upper-triangular entries of $\Gamma^g$ correspond to $\mu^g$.
\end{definition}
The standard covariance matrix characterization~\cite[Section~IV]{Bravyi} gives
\begin{equation}
 (\Gamma^g)^2=-I_d,\quad \Gamma^g(\Gamma^g)^{\trans}=I_d,\quad
 \|\mu^g\|^2=\tfrac12\|\Gamma^g\|_F^2=n.
 \label{eq:gamma-pure}
\end{equation}
Conversely, every real skew matrix $\Gamma$ with $\Gamma^2=-I_d$ is a pure
Gaussian covariance matrix. In the rotated Majorana basis from
\Cref{def:pure-gaussian}, the covariance matrix $O\Gamma^gO^{\trans}$
is block diagonal with all blocks
$\bigl(\begin{smallmatrix}0&-1\\1&0\end{smallmatrix}\bigr)$.

Wick's theorem lets us compute all moments of $\ket{g}$ from its covariance matrix.
We need only its quartic case~\cite[Section~IV, equation~(17)]{Bravyi},
extended to account for repeated indices using the CAR: for any $i,j,k,\ell\in[d]$,
\begin{equation}
 \langle c_ic_jc_kc_\ell\rangle_g
 =\langle c_ic_j\rangle_g\langle c_kc_\ell\rangle_g
  -\langle c_ic_k\rangle_g\langle c_jc_\ell\rangle_g
  +\langle c_ic_\ell\rangle_g\langle c_jc_k\rangle_g.
 \label{eq:wick-four}
\end{equation}
Any Gaussian state $\ket{g}$ we produce is represented using its quadratic moments $\mu^g$. The above formulas suffice to compute its quartic moments, and hence the energies, used here, in polynomial time. This representation allows for computing a broad range of properties of $\ket{g}$ and also preparing $\ket{g}$ in polynomial-time on a quantum computer.

\begin{lemma}[Gaussian variance bound]\label{lem:uniform-variance}
For a Gaussian $\ket{g}$, let $V^g=M^g-\mu^g(\mu^g)^{\trans}$. Then
\begin{equation*}
 0\preceq V^g\preceq2I_m,
\end{equation*}
and in particular
\begin{equation}
 \Var_g(Q_{\vx})=\vx^{\trans}V^g\vx\le2\|\vx\|^2.
 \label{eq:uniform-variance}
\end{equation}
\end{lemma}
\begin{proof}
Take $\Gamma=\Gamma^g$ for convenience. Substituting \eqref{eq:two-point} into
\eqref{eq:wick-four} gives
\begin{equation*}
\begin{aligned}
 \Re\langle c_ic_jc_kc_\ell\rangle_g 
 &=\delta_{ij}\delta_{k\ell}-\delta_{ik}\delta_{j\ell}
   +\delta_{i\ell}\delta_{jk}\\
 &\quad-\Gamma_{ij}\Gamma_{k\ell}
   +\Gamma_{ik}\Gamma_{j\ell}
   -\Gamma_{i\ell}\Gamma_{jk},
\end{aligned}
\end{equation*}
since $\Re (\delta_{st}-i\Gamma_{st})(\delta_{uv}-i\Gamma_{uv}) =\delta_{st}\delta_{uv} - \Gamma_{st}\Gamma_{uv}$. For $i<j$ and $k<\ell$, the above and \eqref{eq:moment-entry} yield
\begin{equation*}
  V^g_{(i,j),(k,\ell)} = 
  -\Re\langle c_ic_jc_kc_\ell\rangle_g - \mu^g_{ij}\mu^g_{k\ell} 
  =\delta_{ik}\delta_{j\ell}-\Gamma_{ik}\Gamma_{j\ell}
   +\Gamma_{i\ell}\Gamma_{jk},
\end{equation*}
since $\mu^g_{ij}\mu^g_{k\ell} = \Gamma_{ij}\Gamma_{k\ell}$ and $\delta_{ij}\delta_{k\ell}=\delta_{i\ell}\delta_{jk}=0$. Thus, for
$\vx\in\mathbb R^{\pairs}$,
\begin{equation}
 \vx^{\trans}V^g\vx
 =\|\vx\|^2-
   \sum_{i<j}\sum_{k<\ell}x_{ij}x_{k\ell}
       (\Gamma_{ik}\Gamma_{j\ell}-\Gamma_{i\ell}\Gamma_{jk}).
 \label{eq:gaussian-variance-form}
\end{equation}
Let $K\in\mathbb R^{d\times d}$ be skew-symmetric with $K_{ij}=x_{ij}$
for $i<j$. Then $\|K\|_F^2=2\|\vx\|^2$. Pairing the $(k,\ell)$
and $(\ell,k)$ terms and using the skew-symmetry of $\Gamma$ gives
\begin{equation*}
 (\Gamma^{\trans} K\Gamma)_{ij}
 = (\Gamma_i)^{\trans} K (\Gamma_j) = \sum_{k<\ell}x_{k\ell}
       (\Gamma_{ik}\Gamma_{j\ell}-\Gamma_{i\ell}\Gamma_{jk}).
\end{equation*}
Both $K$ and $\Gamma^{\trans}K\Gamma$ are skew-symmetric, so their
Frobenius inner product counts each upper-triangular entry twice. Hence
\begin{equation*}
 \sum_{i<j}\sum_{k<\ell}x_{ij}x_{k\ell}
       (\Gamma_{ik}\Gamma_{j\ell}-\Gamma_{i\ell}\Gamma_{jk})
 =\tfrac12\langle K,\Gamma^{\trans} K\Gamma\rangle_F,
\end{equation*}
where $\langle U,W\rangle_F=\Tr[U^{\trans}W]$.
By \eqref{eq:gamma-pure}, $\Gamma$ is orthogonal, so Cauchy--Schwarz gives
\begin{equation*}
 \tfrac12\bigl|\langle K,\Gamma^{\trans}K\Gamma\rangle_F\bigr|
 \le\tfrac12\|K\|_F\|\Gamma^{\trans}K\Gamma\|_F
 =\tfrac12\|K\|_F^2=\|\vx\|^2.
\end{equation*}
Equation~\eqref{eq:gaussian-variance-form} therefore gives
$\vx^{\trans}V^g\vx\le2\|\vx\|^2$.
The lower bound follows from \eqref{eq:moment-psd}.
\end{proof}

We will need the Gaussian cone $\mathcal C_G$ to be closed under
certain operations. An operator is \emph{even} if its Majorana expansion
contains only even-degree monomials. The following lemma also applies to mixed Gaussian states, but we use it only for pure Gaussian states.

\begin{lemma}[Gaussian closure {\cite[Theorem~1 and Corollary~2]{Bravyi}}]\label{lem:gaussian-operator-criterion}
Let $\Lambda=\sum_{k=1}^d c_k\otimes c_k$ act on
$\mathcal H\otimes\mathcal H$. If an even operator $X$ satisfies
$[\Lambda,X\otimes X]=0$, then for every Gaussian density matrix $\rho$
with $X\rho X^\dagger\ne0$, the normalized density matrix
\begin{equation*}
 \frac{X\rho X^\dagger}{\Tr[X\rho X^\dagger]}
\end{equation*}
is Gaussian.
\end{lemma}

\begin{corollary}\label{lem:gaussian-closure}
For every $\vx\in\mathbb R^{\pairs}$ and $z\in\mathbb C$,
$e^{zQ_{\vx}}$ maps $\mathcal C_G$ into itself.
\end{corollary}
\begin{proof}
Take $Q=Q_{\vx}$ and $A=e^{zQ}$. For $i < j$, the CAR give
\begin{equation*}
 [c_k, c_ic_j]=2(\delta_{ik}c_j-\delta_{jk}c_i).
\end{equation*}
Consequently,
\begin{equation*}
 [(c_i\otimes c_i + c_j\otimes c_j),(c_ic_j\otimes I + I \otimes c_ic_j)]=0,
\end{equation*}
so $[\Lambda,Q\otimes I+I\otimes Q]=0$. Exponentiating yields
$[\Lambda,A\otimes A]=0$, since
$A\otimes A=e^{z(Q\otimes I+I\otimes Q)}$.
Also, $A$ is even because $Q$ is quadratic. For a Gaussian
$\ket{g}$, set $\rho=\ket{g}\bra{g}$. The operator $A$ is
invertible, with inverse $e^{-zQ}$, so $A\ket{g}\ne0$ and $A\rho A^\dagger=A\ket{g}\bra{g}A^\dagger$ has rank one. By \Cref{lem:gaussian-operator-criterion}, dividing by
its trace $\|A\ket{g}\|^2$ gives the Gaussian projector onto
$A\ket{g}/\|A\ket{g}\|$. Thus $A\ket{g}\in\mathcal C_G$, and linearity gives the claim
for every scalar multiple of $\ket{g}$, including zero.
\end{proof}

\section{Approximating the moment matrix of a state}\label{sec:ensemble}

This section shows that the moment matrix of a state $\rho$ can be approximated with the moment matrix of a mixture $\omega$ of Gaussian states:
\begin{equation*}
\|M^\rho-M^\omega\|_\op\le2,
\end{equation*}
proven in \Cref{prop:ensemble}. We start by showing that for a state $\ket{\psi}$, choosing a Gaussian $\ket{g}$ with maximum overlap with $\ket{\psi}$ provides a good approximation. However, the approximation is in a different sense than above.

\begin{lemma}\label{lem:local-overlap}
Let $\ket{\psi}$ be a state and $\ket{g}$ be a Gaussian maximizing $|\langle g|\psi\rangle|$
over all Gaussians. The maximum is positive because
Gaussian Fock basis states span $\mathcal H$. Choose the phase of $\ket{g}$
so that $a:=\langle g|\psi\rangle>0$, and define the unnormalized residual $\ket{\xi}=\ket{\psi}-a\ket{g}$.
For every quadratic Hamiltonian $Q=Q_{\vx}$,
\begin{equation}
 \langle\xi|Q_{\vx}|g\rangle=0,
 \quad
 |\langle\xi|Q_{\vx}^2|g\rangle|\le a\Var_g(Q_{\vx})\le2a\|\vx\|^2.
 \label{eq:local-residual}
\end{equation}
The second bound is equivalently
\begin{equation}
 \left|\frac{\langle g|Q_{\vx}^2|\psi\rangle}{a}
                  -\langle Q_{\vx}^2\rangle_g\right|
 \le\Var_g(Q_{\vx})\le2\|\vx\|^2.
 \label{eq:local-overlap}
\end{equation}
\end{lemma}
\begin{proof}
We first show that $a\ket{g}$ is a closest element of $\mathcal C_G$
to $\ket{\psi}$. Every element of $\mathcal C_G$ has the form
$\alpha\ket{g'}$, where $\alpha\in\mathbb C$ and $\ket{g'}$ is a
Gaussian, including zero by taking $\alpha=0$.
By maximality, $|\langle g'|\psi\rangle|\le a$, so
\begin{equation*}
 \begin{aligned}
 \|\ket{\psi}-\alpha\ket{g'}\|^2
 &=1+|\alpha|^2-2\Re(\overline\alpha\langle g'|\psi\rangle)\\
 &\ge1+|\alpha|^2-2a|\alpha|\\
 &=1-a^2+(|\alpha|-a)^2\\
 &\ge1-a^2=\|\ket{\psi}-a\ket{g}\|^2.
 \end{aligned}
\end{equation*}
Also, since $\ket{g}$ is normalized,
\begin{equation*}
 \langle g|\xi\rangle
 =\langle g|\psi\rangle-a\langle g|g\rangle=0.
\end{equation*}

Fix $Q$ and set
\begin{equation*}
 \mu_Q=\langle Q\rangle_g,\quad
 \widehat Q=Q-\mu_Q I,\quad
 \nu_Q=\|\widehat Q\ket{g}\|^2=\Var_g(Q).
\end{equation*}
Subtracting $\mu_Q I$ makes $\widehat Q\ket{g}$ orthogonal to $\ket{g}$.
For $z\in\mathbb C$, define the unnormalized Gaussian vector and loss
\begin{equation}
 \ket{\phi_Q(z)}=a e^{z\widehat Q}\ket{g},\quad
 \ell_Q(z)=\tfrac12\|\ket{\psi}-\ket{\phi_Q(z)}\|^2.
 \label{eq:local-objective}
\end{equation}
Since $e^{z\widehat Q}=e^{-z\mu_Q}e^{zQ}$,
\Cref{lem:gaussian-closure} gives $\ket{\phi_Q(z)}\in\mathcal C_G$.
Also $\ket{\phi_Q(0)}=a\ket{g}$, so for every $u\in\mathbb C$ the function $t\mapsto\ell_Q(tu)$, $t\in\mathbb R$, has a global
minimum at $t=0$. Fix $u\in\mathbb C$ and define the residual vector
$\ket{r(t)}=\ket{\psi}-\ket{\phi_Q(tu)}$. This residual and its derivatives
need not be normalized. We have
$\ell_Q(tu)=\tfrac12\langle r(t)|r(t)\rangle$.
With primes denoting derivatives with respect to $t$, \eqref{eq:local-objective} gives
\begin{equation*}
 \ket{r'(t)}=-u\widehat Q\ket{\phi_Q(tu)},\quad
 \ket{r''(t)}=-u^2\widehat Q^2\ket{\phi_Q(tu)}.
\end{equation*}
In particular,
\begin{equation}
 \ket{r(0)}=\ket{\xi},\quad
 \ket{r'(0)}=-au\widehat Q\ket{g},\quad
 \ket{r''(0)}=-au^2\widehat Q^2\ket{g}.
 \label{eq:local-residual-at-zero}
\end{equation}
The product rule yields
\begin{align}
 \frac{d}{dt}\ell_Q(tu)
 &=\tfrac12\bigl(\langle r'(t)|r(t)\rangle
                  +\langle r(t)|r'(t)\rangle\bigr)
 =\Re\langle r(t)|r'(t)\rangle,
 \label{eq:local-loss-first-derivative}\\
 \frac{d^2}{dt^2}\ell_Q(tu)
 &=\|\ket{r'(t)}\|^2+\Re\langle r(t)|r''(t)\rangle.
 \label{eq:local-loss-second-derivative}
\end{align}
Put $\eta_Q=\langle\xi|\widehat Q|g\rangle$. The minimum at $t=0$
implies that the first derivative vanishes, so
\eqref{eq:local-residual-at-zero} and
\eqref{eq:local-loss-first-derivative} give
\begin{equation}
 0=\left.\frac{d}{dt}\ell_Q(tu)\right|_{t=0}
 =-a\Re(u\eta_Q).
 \label{eq:local-first-derivative}
\end{equation}
Taking $u=1$ and $u=i$ in \eqref{eq:local-first-derivative} gives
$\eta_Q=0$. Since $\langle\xi|g\rangle=0$, we also obtain
$\langle\xi|Q|g\rangle=0$.

Let $\zeta_Q=\langle\xi|Q^2|g\rangle$. Expanding
$\widehat Q^2=Q^2-2\mu_QQ+\mu_Q^2I$ and using
$\langle\xi|g\rangle=\langle\xi|Q|g\rangle=0$ gives
$\langle\xi|\widehat Q^2|g\rangle=\zeta_Q$.
The second derivative is nonnegative at the minimum $t=0$, so
\eqref{eq:local-loss-second-derivative} and
\eqref{eq:local-residual-at-zero} give
\begin{equation*}
 \begin{aligned}
 0\le\left.\frac{d^2}{dt^2}\ell_Q(tu)\right|_{t=0}
 &=\|au\widehat Q\ket{g}\|^2
   -a\Re\bigl(u^2\langle\xi|\widehat Q^2|g\rangle\bigr)\\
 &=a^2|u|^2\nu_Q-a\Re(u^2\zeta_Q).
 \end{aligned}
\end{equation*}
Choose $u$ with $|u|=1$ so that $u^2\zeta_Q=|\zeta_Q|$.
Then $|\zeta_Q|\le a\nu_Q$.
\Cref{eq:uniform-variance} gives $\nu_Q\le2\|\vx\|^2$ and proves
\eqref{eq:local-residual}. Finally, since $Q^2$ is Hermitian,
\begin{equation*}
 \langle g|Q^2|\psi\rangle
 =a\langle Q^2\rangle_g+\overline{\zeta_Q},
\end{equation*}
which gives \eqref{eq:local-overlap} after dividing by $a$.
\end{proof}

\paragraph{Establishing a coefficient-matrix-dependent bound.}
One can see how the bound \eqref{eq:local-overlap} might become more useful for our purposes by taking $\ket{\psi}$ to be a top eigenvector of $Q^2$, corresponding to its maximum eigenvalue. In this case, for $\ket{g}$ and $a$ chosen as in
\Cref{lem:local-overlap}, we get
\begin{equation*}
 \lambda_{\max}(Q^2) = \frac{\langle g|Q^2|\psi\rangle}{a},
\end{equation*}
so that
\begin{equation*}
 \lambda_{\max}(Q^2)-G(Q^2)
 \le\frac{\langle g|Q^2|\psi\rangle}{a}-\langle Q^2\rangle_g
 \le \Var_g(Q).
\end{equation*}

We use this idea to next derive a Gaussian moment-matrix approximation of $M^\rho$ that depends on a coefficient-matrix $A$. Let $\|A\|_*$ denote the nuclear norm of $A$, the sum of the absolute values
of its eigenvalues.

\begin{lemma}\label{lem:scalar-comparison}
For every density matrix $\rho$ and real symmetric coefficient matrix $A$,
\begin{equation}
 \Tr[AM^\rho]\le
 \max_{\ket{g}\text{Gaussian}}\Tr[AM^g]+2\|A\|_*.
 \label{eq:support-comparison}
\end{equation}
\end{lemma}
\begin{proof}
The spectral decomposition of $A$ gives
\begin{equation*}
 A=\sum_r\alpha_r\vx_r\vx_r^{\trans},\quad
 \|\vx_r\|=1,\quad
 H_A=\sum_r\alpha_r Q_{\vx_r}^2.
\end{equation*}
Here $\alpha_r\in\mathbb R$ and the $\vx_r$ form an orthonormal eigenbasis.
Choose a top eigenstate $\ket{\psi}$ of $H_A$, and a Gaussian $\ket{g}$ maximizing $a=\langle g|\psi\rangle>0$ as in the hypothesis of \Cref{lem:local-overlap}. Applying \eqref{eq:local-overlap} to each term of the above decomposition of $H_A$ gives
\begin{align}
 \lambda_{\max}(H_A)
 &=\frac{\langle g|H_A|\psi\rangle}{a}\notag\\
 &\le\sum_r\alpha_r\langle Q_{\vx_r}^2\rangle_g
                         +2\sum_r|\alpha_r|
 =\Tr[AM^g]+2\|A\|_*,
 \label{eq:support-bound}
\end{align}
where the last equality follows from \eqref{eq:moment-SoS-duality}; the latter also gives that for every density matrix $\rho$, $\Tr[AM^\rho]=\langle H_A\rangle_\rho\le\lambda_{\max}(H_A)$.
Maximizing the right side of \eqref{eq:support-bound} over $\ket{g}$ proves
\eqref{eq:support-comparison}.
\end{proof}

In the lemma, since $\ket{\psi}$ depends on $A$, so does $\ket{g}$. We show how this dependence can be removed by taking a mixture of Gaussians instead of a single Gaussian $\ket{g}$.

\begin{proposition}\label{prop:ensemble}
For every density matrix $\rho$, there is a mixture $\omega$ of at most
$\binom d4+1$ pure Gaussian states such that
\begin{equation}
 M^\omega-2I_m\preceq M^\rho\preceq M^\omega+2I_m.
 \label{eq:uniform-moment}
\end{equation}
Equivalently, 
\begin{equation*}
\|M^\rho-M^\omega\|_\op\le2.
\end{equation*}
\end{proposition}
\begin{proof}
Let
\begin{equation*}
 \mathcal M_G=\conv\{M^g:\ket{g}\text{ Gaussian}\},\quad
 \mathcal A=\{A\in\mathbb R^{m\times m}:A=A^{\trans},\ \|A\|_\op\le1\}.
\end{equation*}
The Gaussian projectors form a compact set, and their moment matrices
depend continuously on them. The convex hull of a compact set in a
finite-dimensional space is compact, so $\mathcal M_G$ is compact and
convex. The same holds for $\mathcal A$ and hence for
$\mathcal M_G+2\mathcal A=\{M+2A:M\in\mathcal M_G,\ A\in\mathcal A\}$.

For every real symmetric matrix $B$, duality between the operator and nuclear norms gives
\begin{equation}
 \max_{A\in\mathcal A}\Tr[BA]=\|B\|_*.
 \label{eq:operator-ball-support}
\end{equation}
Indeed, decompose $B=\sum_r\beta_r\vx_r\vx_r^{\trans}$ in an orthonormal
eigenbasis. For $A\in\mathcal A$, we have
$|\vx_r^{\trans}A\vx_r|\le1$, so
$\Tr[BA]\le\sum_r|\beta_r|=\|B\|_*$, with equality at
$A=\sum_r\operatorname{sgn}(\beta_r)\vx_r\vx_r^{\trans}$.

Suppose $M^\rho\notin\mathcal M_G+2\mathcal A$.
Strict separation, with the trace inner product on the real vector
space of symmetric $m\times m$ matrices, gives a real symmetric $B$ such that
\begin{equation}
 \begin{aligned}
 \Tr[BM^\rho]
 &>\max_{M\in\mathcal M_G}\Tr[BM]
       +2\max_{A\in\mathcal A}\Tr[BA]\\
 &=\max_{\ket{g}\text{ Gaussian}}\Tr[BM^g]+2\|B\|_*.
 \end{aligned}
 \label{eq:separation-contradiction}
\end{equation}
The equality uses \eqref{eq:operator-ball-support} and the fact that a
linear functional has the same maximum over a set and its convex hull.

Since \eqref{eq:separation-contradiction} contradicts
\eqref{eq:support-comparison} in \Cref{lem:scalar-comparison},
$M^\rho=M+2A$ for some $M\in\mathcal M_G$ and $A\in\mathcal A$,
so $\|M^\rho-M\|_\op\le2$.
By \eqref{eq:moment-entry}, the matrices in $\mathcal M_G$ lie in an affine
space of dimension at most $\binom d4$. Carath\'eodory's theorem therefore
expresses $M=M^\omega$ using at most $\binom d4+1$ Gaussian states,
proving \eqref{eq:uniform-moment}.
\end{proof}

\paragraph{Positive semidefinite coefficient matrices.}
For PSD coefficient matrices $A$, corresponding to $H_A$ that are sums of squares of quadratic Hamiltonians (see the end of \Cref{sec:pairs}), our results in this section show that there exists a Gaussian state earning a constant fraction of the maximum energy.

\begin{corollary}\label{cor:psd}
For every real symmetric $A\succeq0$,
\begin{equation*}
 G(H_A)\ge\tfrac12\lambda_{\max}(H_A).
\end{equation*}
\end{corollary}
\begin{proof}
Factor $A=\sum_j\vx_j\vx_j^{\trans}$,
so $H_A=\sum_jQ_{\vx_j}^2$.
Let $\ket{\psi}$ be a top eigenstate of $H_A$, and choose a Gaussian $\ket{g}$ maximizing $a=\langle g|\psi\rangle>0$ and $\ket{\xi}=\ket{\psi}-a\ket{g}$ as in the hypothesis of \Cref{lem:local-overlap}.
\begin{align*}
 \lambda_{\max}(H_A)
 &=\langle H_A\rangle_g+\frac{\langle g|H_A|\xi\rangle}{a}\\
 &\le\langle H_A\rangle_g+\sum_j\Var_g(Q_{\vx_j})\\
 &=2\langle H_A\rangle_g-\sum_j\langle Q_{\vx_j}\rangle_g^2
 \le2G(H_A),
\end{align*}
where the first inequality is an application of \eqref{eq:local-residual}.
\end{proof}

\section{Bounding the maximum eigenvalue}\label{sec:coefficient}
This section proves our main eigenvalue bound
\begin{equation*}
 \lambda_{\max}(H)\le G(H)+5\sqrt n\,\sigma.
\end{equation*}
We use the moment-matrix approximation
from \Cref{prop:ensemble},
\begin{equation*}
\|M^\rho-M^\omega\|_\op\le2,
\end{equation*}
to derive the Frobenius-norm bound
\begin{equation*}
\|M^\rho-M^\omega\|_F\le5\sqrt{6n},
\end{equation*}
from which the eigenvalue bound follows after an application of Cauchy--Schwarz.

The next lemma bounds $R^2$ for a homogeneous quartic $R$.
The identity \eqref{eq:contraction-count} will also be used to prove
\Cref{lem:sparse-state} in the next section.

\begin{lemma}\label{lem:psd-lift}
For every homogeneous quartic $R$, let $C=C^R$ and $s=\sqrt{\tau[R^2]}$.
Then
\begin{equation}
 H_{C^2}=6s^2I+[R^2]_4,
 \label{eq:contraction-count}
\end{equation}
and
\begin{equation}
 R^2\preceq\frac n4 H_{C^2}-ns^2I.
 \label{eq:psd-surrogate}
\end{equation}
\end{lemma}
\begin{proof}
Let $R=\sum_{S\subseteq[d]:|S|=4}r_Sc_S$.
By \Cref{def:moments} and the symmetry of $C$,
\begin{equation*}
 H_{C^2}=\boldsymbol q^{\trans}C^{\trans}C\boldsymbol q
 =(C\boldsymbol q)^{\trans}(C\boldsymbol q)
 =\sum_{p\in E}(C\boldsymbol q)_p^2.
\end{equation*}
For $p\in E$, $(C\boldsymbol q)_p^2=(\sum_{r\in E}C_{pr}q_r)^2$
has only identity and quartic
terms: if distinct $r,t\in E$ share one index, then $q_r$ and $q_t$
anticommute, so their two ordered products cancel. The identity part is
$\tau[H_{C^2}]I=\Tr[C^2]I=6s^2I$ by \eqref{eq:quartic-coefficient-norm}.

For $S,T\subseteq[d]$ with $|S|=|T|=4$, the product $c_Sc_T$ has degree
four exactly when $|S\cap T|=2$. Let
$S\cap T=\{i,j\}$, $S\setminus T=\{k,\ell\}$, and
$T\setminus S=\{u,v\}$, with $i<j$, $k<\ell$, and $u<v$.
Set $p=(i,j)$, $r=(k,\ell)$, and $t=(u,v)$ in $E$.
By \eqref{eq:quartic-trace-pairing}, using $q_p^2 = I$ and the pairwise commutativity of $q_p$, $q_r$, and $q_t$, we obtain
\begin{equation*}
 r_Sc_S=C_{pr}q_pq_r,\quad r_Tc_T=C_{pt}q_pq_t,
 \quad
 (r_Sc_S)(r_Tc_T)=C_{pr}C_{pt}q_rq_t.
\end{equation*}
Thus each ordered degree-four product in $R^2$ occurs, with the same
sign and multiplicity, in exactly one summand $(C\boldsymbol q)_p^2$.
Conversely, a quartic term in $(C\boldsymbol q)_p^2$ comes from $r,t\in E$ with disjoint
index sets; a nonzero coefficient requires both to be disjoint from the
index set of $p$. These three index sets recover $S,T$ as above.
This proves \eqref{eq:contraction-count}.

To prove the bound on $R^2$ in \eqref{eq:psd-surrogate}, we average
over conjugation by each Majorana:
\begin{equation*}
 \mathcal E(X)=\frac1d\sum_{i=1}^d c_iXc_i.
\end{equation*}
For a Majorana monomial $c_S$ of even degree, conjugation by $c_i$
changes its sign exactly when $i\in S$. Therefore
\begin{equation}
 \mathcal E(c_S)=\left(1-\frac{2|S|}{d}\right)c_S
 \quad(S\subseteq[d],\ |S|\text{ even}).
 \label{eq:conjugation-degree}
\end{equation}
For every Hermitian operator $X$ on $\mathcal H$,
\begin{equation}
 \mathcal E(X^2)-\mathcal E(X)^2
 =\frac1d\sum_{i=1}^d\bigl(c_iXc_i-\mathcal E(X)\bigr)^2\succeq0.
 \label{eq:matrix-variance}
\end{equation}
To verify \eqref{eq:matrix-variance}, set $B=\mathcal E(X)$ and expand: \begin{equation*}  \frac1d\sum_{i=1}^d(c_iXc_i-B)^2  =\mathcal E(X^2)-\mathcal E(X)B-B\mathcal E(X)+B^2  =\mathcal E(X^2)-B^2. \end{equation*} Here $c_i^2=I$ and $\mathcal E(X)=B$; each squared operator is PSD because $c_iXc_i-B$ is Hermitian.

Set \(U=[R^2]_4\). Since \(R\) has degree four and \(R^2\) has only
degrees zero, four, and eight, \eqref{eq:conjugation-degree} gives
\begin{equation*}
 \mathcal E(R)=\left(1-\frac8d\right)R,\quad
 \mathcal E(R^2)=\left(1-\frac{16}{d}\right)R^2
                 +\frac8d U+\frac{16}{d}s^2I.
\end{equation*}
By \eqref{eq:matrix-variance}, we have
\begin{equation*}
 0\preceq\mathcal E(R^2)-\mathcal E(R)^2
 =\frac8d(U+2s^2I)-\frac{64}{d^2}R^2.
\end{equation*}
By \eqref{eq:contraction-count}, \(U+2s^2I=H_{C^2}-4s^2I\).
Multiplying by \(d^2/64\) and using \(d=2n\) proves
\eqref{eq:psd-surrogate}.
\end{proof}

We use the previous result to obtain a bound on $\|R\|_\op$, for a quartic homogeneous $R$, from one on the operator norm of its coefficient matrix, $\|C^R\|_\op$.

\begin{lemma}\label{lem:normalized-spectral}
If $R$ is homogeneous quartic, $\tau[R^2]=1$, and
$\|C^R\|_\op\le \tfrac{1}{\sqrt n}$, then
\begin{equation*}
 \|R\|_\op\le\tfrac52\sqrt n.
\end{equation*}
\end{lemma}
\begin{proof}
For every real symmetric $A\in\mathbb R^{m\times m}$ with $A\succeq0$,
\begin{equation}
 \lambda_{\max}(H_A)\le n\|A\|_\op+4\Tr[A].
 \label{eq:psd-spectral}
\end{equation}
Indeed, \eqref{eq:support-comparison} applied to a top eigenstate gives
$\lambda_{\max}(H_A)\le G(H_A)+2\Tr[A]$ since $\Tr[A] = \|A\|_*$. For a Gaussian $\ket{g}$,
\eqref{eq:uniform-variance} and $\|\mu^g\|^2=n$ give
\begin{equation*}
 \langle H_A\rangle_g
 =\Tr[AM^g]
 \le(\mu^g)^{\trans}A\mu^g+2\Tr[A]
 \le n\|A\|_\op+2\Tr[A].
\end{equation*}
Taking the maximum over $\ket{g}$ proves \eqref{eq:psd-spectral}.

Now take $C=C^R$ and set $A=C^2\succeq0$. Then $\Tr[A]=6$ by
\eqref{eq:quartic-coefficient-norm}, and $\|A\|_\op\le1/n$. Hence
\eqref{eq:psd-spectral} gives $\lambda_{\max}(H_A)\le1+24=25$.
The bound \eqref{eq:psd-surrogate} implies $R^2\preceq(n/4)H_A$, so
\begin{equation*}
 \|R\|_\op^2=\lambda_{\max}(R^2)
 \le\tfrac n4\lambda_{\max}(H_A)\le\tfrac{25n}{4}.
\end{equation*}
\end{proof}

We next convert an operator-norm bound on the difference of two moment
matrices into a Frobenius-norm bound. Special properties of moment
matrices arising from physical states afford an improvement over more generic relationships between the two norms.

\begin{lemma}\label{prop:norm-conversion}
For density matrices $\rho_1,\rho_2$,
\begin{equation*}
 \|M^{\rho_1}-M^{\rho_2}\|_\op\le2
 \implies
 \|M^{\rho_1}-M^{\rho_2}\|_F\le5\sqrt{6n}.
\end{equation*}
\end{lemma}
\begin{proof}
Let $D=M^{\rho_1}-M^{\rho_2}$ and $\Delta=\|D\|_F/\sqrt6$.
Suppose for a contradiction that $\Delta>5\sqrt n$.
Define the homogeneous quartic operator
\begin{equation*}
 R=\frac1\Delta\sum_{S:|S|=4}
 \bigl(\langle c_S\rangle_{\rho_1}-\langle c_S\rangle_{\rho_2}\bigr)c_S.
\end{equation*}
By \eqref{eq:moment-entry}, $C^R=D/\Delta$, so
by \eqref{eq:quartic-coefficient-norm} and the hypothesis of the lemma, respectively,
\begin{equation*}
 \tau[R^2]=\frac{\|C^R\|_F^2}{6}
 =\frac{\|D\|_F^2}{6\Delta^2}=1, \quad
 \|C^R\|_\op\le\frac2\Delta<\frac1{\sqrt n}.
\end{equation*}
\Cref{lem:normalized-spectral} therefore bounds
$\|R\|_\op\le(5/2)\sqrt n$. By \eqref{eq:moment-SoS-duality} and \eqref{eq:quartic-reconstruction},
\begin{equation*}
 \langle R\rangle_{\rho_1}-\langle R\rangle_{\rho_2}
 =\frac{\Tr[D(M^{\rho_1}-M^{\rho_2})]}{6\Delta}
 =\frac{\|D\|_F^2}{6\Delta}=\Delta.
\end{equation*}
Thus
\begin{equation*}
 \Delta\le\lambda_{\max}(R)-\lambda_{\min}(R)
 \le2\|R\|_\op\le5\sqrt n,
\end{equation*}
contradicting the assumption.
\end{proof}

We now obtain the eigenvalue bound by choosing one of the two states in
\Cref{prop:norm-conversion} to be a Gaussian mixture.

\begin{mainbound}[restated]
For every homogeneous quartic Hamiltonian $H$ on $n\ge2$ modes,
\begin{equation*}
 \lambda_{\max}(H)\le G(H)+5\sqrt n\,\sigma.
\end{equation*}
\end{mainbound}
\begin{proof}
Let $\ket{\psi}$ be a top eigenstate of $H$ with $\rho = \ket{\psi}\bra{\psi}$, and choose $\omega$ as in
\Cref{prop:ensemble}. \Cref{prop:norm-conversion}
gives $\|M^\rho-M^\omega\|_F\le5\sqrt{6n}$. Since $\omega$ is a mixture of
Gaussians, $\langle H\rangle_\omega\le G(H)$. Also,
$H_{C^H}=6H$ by \eqref{eq:quartic-reconstruction}, and
$\|C^H\|_F=\sqrt6\,\sigma$ by \eqref{eq:quartic-coefficient-norm}.
Thus, by \eqref{eq:moment-SoS-duality}
and Cauchy--Schwarz,
\begin{align*}
 \lambda_{\max}(H)-G(H)
 &\le\langle H\rangle_\rho-\langle H\rangle_\omega
   =\frac16\Tr[C^H(M^\rho-M^\omega)]\\
 &\le\frac16\|C^H\|_F\|M^\rho-M^\omega\|_F
 \le\frac16(\sqrt6\,\sigma)(5\sqrt{6n})
 =5\sqrt n\,\sigma.
\end{align*}
\end{proof}

\section{Constructing a good candidate state}
\label{sec:filter}
This section shows how to construct a state yielding a constant-factor approximation with respect to $\sigma$. We apply $16\sigma I+H$
to a uniformly random Fock basis state and test whether the resulting
normalized state has energy greater than $\sigma/16$. The following cubic trace bound is used to show that polynomially many trials suffice.

\begin{lemma}\label{lem:cubic-trace}
For every homogeneous quartic $R$, let $C=C^R$ and $s=\sqrt{\tau[R^2]}$.
Then
\begin{equation*}
 \tau[R^3]=\Tr[C^3],
\end{equation*}
and
\begin{equation*}
 |\tau[R^3]|\le6^{3/2}s^3\le15s^3,
\end{equation*}
where the last inequality is strict if $s>0$.
\end{lemma}
\begin{proof}
If $s=0$, then $R=0$ and both assertions are immediate. Assume $s>0$.
Since $R$ is homogeneous quartic,
\eqref{eq:orthogonal-words} gives
\begin{equation*}
 \tau[R^3]=\tau[R[R^2]_4],
\end{equation*}
because all non-quartic components of $R^2$ are orthogonal to $R$.
By \eqref{eq:contraction-count},
\begin{equation*}
 \tau[R[R^2]_4]=\tau[RH_{C^2}]-6s^2\tau[R]=\tau[RH_{C^2}].
\end{equation*}
Finally, expanding $H_{C^2}$ using \eqref{eq:SA} and applying
\eqref{eq:quartic-trace-pairing} yields
\begin{equation*}
 \tau[RH_{C^2}]
 =\sum_{p,r\in E}(C^2)_{pr}\tau[Rq_pq_r]
 =\sum_{p,r\in E}(C^2)_{pr}C_{pr}
 =\Tr[C^3],
\end{equation*}
where the last equality uses the symmetry of $C$.
If $\lambda_j(C)$ are the eigenvalues of $C$, then
\begin{equation*}
 |\tau[R^3]|\le\sum_j|\lambda_j(C)|^3
 \le\left(\sum_j\lambda_j(C)^2\right)^{3/2}
 =(6s^2)^{3/2}<15s^3,
\end{equation*}
where for the middle inequality, we bound each $|\lambda_j(C)|$ by
$(\sum_j\lambda_j(C)^2)^{1/2}$. The equality
$\sum_j\lambda_j(C)^2=6s^2$ follows from \eqref{eq:quartic-coefficient-norm}.
\end{proof}

\begin{sparsestate}[restated]
There is a randomized polynomial-time classical algorithm that, given
a Hamiltonian $H$ of the form \eqref{eq:input} on $n\ge2$ modes, returns,
with probability at least $5/6$, a description of a state $\ket{\psi_F}$
satisfying
\begin{equation*}
 \langle H\rangle_{\psi_F}\ge\sigma/16.
\end{equation*}
The state is supported on at most
\begin{equation*}
 \min\left\{t+1,\ 1+\binom n2+\binom n4\right\}
\end{equation*}
Fock basis strings. For $H\ne0$, it is proportional to
$(16\sigma I+H)\ket{\vy}$ for some $\vy\in\{0,1\}^n$.
Its energy is efficiently computable classically.
The running time is polynomial in $n$ and the input length.
\end{sparsestate}
\begin{proof}
If $H=0$, return any Fock basis state. Otherwise, let
$F=16I+H/\sigma$. Sample $\vy$ uniformly from $\{0,1\}^n$ and form the
unnormalized vector $\ket{v_{\vy}}=F\ket{\vy}$. We accept the trial if
\begin{equation}
 w(\vy):=\bra{v_{\vy}}(H-\sigma I/16)\ket{v_{\vy}}>0.
 \label{eq:witness-polynomial}
\end{equation}
This test ensures that $\ket{v_{\vy}}\ne0$ and that its normalization has
energy greater than $\sigma/16$.

Since $\tau[H]=0$, $\tau[H^2]=\sigma^2$, and by \Cref{lem:cubic-trace}, $\tau[H^3] \geq -15\sigma^3$: 
\begin{equation}
 \mathbb E_{\vy}[w(\vy)]
 =\tau[F(H-\sigma I/16)F]
 =\frac{\tau[H^3]}{\sigma^2}+\frac{255\sigma}{16}
 \ge \frac{15\sigma}{16}.
 \label{eq:witness-positive}
\end{equation}
To bound the acceptance probability, we first bound $w(\vy)$ uniformly.
Each Majorana operator is unitary by \eqref{eq:car}, so every monomial
$c_S$ has operator norm one. The triangle
inequality, followed by Cauchy--Schwarz over the $t$ nonzero coefficients,
therefore gives
\begin{equation*}
 \|H\|_\op\le\sum_{S:w_S\ne0}|w_S|
 \le\sqrt t\left(\sum_S w_S^2\right)^{1/2}
 =\sqrt t\,\sigma,
\end{equation*}
where the last equality is \eqref{eq:scales}. This implies $\|F\|_\op\le16+\sqrt t$. Let $b=\sqrt t(16+\sqrt t)^2$. Dropping the nonpositive term
$-\sigma\|v_{\vy}\|^2/16$ in \eqref{eq:witness-polynomial} and using
$\|v_{\vy}\|\le\|F\|_\op$ gives
\begin{equation*}
 w(\vy)\le\|H\|_\op\|v_{\vy}\|^2
 \le\|H\|_\op\|F\|_\op^2\le\sigma b.
\end{equation*}
Let $p=\Pr[w(\vy)>0]$. On accepted trials, $w(\vy)\le\sigma b$;
on rejected trials, $w(\vy)\le0$. Thus \eqref{eq:witness-positive} implies
\begin{equation*}
 \frac{15\sigma}{16}\le\mathbb E_{\vy}[w(\vy)]\le p\sigma b,
\end{equation*}
so $p\ge15/(16b)$.
After $N=\lceil32b/15\rceil=O(t^{3/2})$ independent trials,
the probability that every trial is rejected is at most
\begin{equation*}
 (1-p)^N\le e^{-Np}\le e^{-2}<1/6.
\end{equation*}

Expanding a quartic Majorana monomial in creation and annihilation
operators shows that its action on a Fock basis string flips zero, two,
or four occupation bits and multiplies the state by an efficiently
computable phase in $\{\pm1,\pm i\}$. Thus $\ket{v_{\vy}}=F\ket{\vy}$
is supported on at most $t+1$ strings. Since only zero, two, or four
bits can differ from $\vy$, its support is also at most
$1+\binom n2+\binom n4$. Applying the $t$ terms of $H$ to this sparse
vector lets us evaluate \eqref{eq:witness-polynomial} and, for an accepted
trial, the energy of its normalization in polynomial time. We output the normalized
vector from the first accepted trial, or report failure if all $N$
trials are rejected.
\end{proof}

\section{Establishing the approximation guarantee}\label{sec:algorithm}
We prove our main result, \Cref{thm:general}, by combining the eigenvalue bound of
\Cref{thm:main}, the sparse state from \Cref{lem:sparse-state}, and a
Gaussian-state approximation algorithm.

Bravyi, Gosset, K\"onig, and Temme~\cite[Theorem~4]{BGKT}
give a randomized polynomial-time classical algorithm that returns a
Gaussian state $\ket{g}$ of energy
\begin{equation}
 \langle H\rangle_g\ge\frac{G(H)}{c_G\log(2n)},
 \label{eq:BGKT}
\end{equation}
with probability at least $2/3$, for a universal constant $c_G$.

\begin{mainresult}[restated]
There is a randomized polynomial-time classical algorithm that, given
a Hamiltonian $H$ of the form \eqref{eq:input} on $n\ge2$ modes, returns
a description of a state $\ket{\psi}$ satisfying
\begin{equation*}
 \langle H\rangle_\psi\ge
 \frac{\lambda_{\max}(H)}{c_G\log(2n)+80\sqrt n}
\end{equation*}
with probability at least $2/3$, for a universal constant $c_G$.
The output is either a fermionic Gaussian state or a superposition
of at most $t+1$ Fock basis states.
Its energy is efficiently computable classically.
The running time is polynomial in $n$ and the input length.
\end{mainresult}
\begin{proof}
If $H=0$, return any Fock basis state. Otherwise, run the algorithm of
\Cref{lem:sparse-state} and run the Gaussian-state approximation algorithm
twice independently, keeping the better Gaussian state $\ket{g}$.
The latter satisfies \eqref{eq:BGKT} with probability at least
$8/9>5/6$. Thus
\begin{equation}\label{eq:sigma-bound}
 \langle H \rangle_{\psi_F} \geq \frac{\sigma}{16}
\end{equation}
and \eqref{eq:BGKT} both hold with probability at least $1-1/6-1/6=2/3$.

Return the higher-energy candidate as $\ket{\psi}$, or return $\ket{g}$
if the sparse-state algorithm reports failure. On the joint success event,
\begin{align*}
 \lambda_{\max}(H)
 &\le G(H)+5\sqrt n\,\sigma\\
 &\le c_G\log(2n)\langle H\rangle_g
       +80\sqrt n\,\langle H\rangle_{\psi_F}\\
 &\le\bigl(c_G\log(2n)+80\sqrt n\bigr)\langle H\rangle_\psi.
\end{align*}
The first inequality is \eqref{eq:main-additive}. The second uses
\eqref{eq:BGKT} and \eqref{eq:sigma-bound}, and the last follows from our choice
of $\ket{\psi}$. The sparse output has support at most $t+1$ by
\Cref{lem:sparse-state}. Both candidate descriptions and their energies
are computable in polynomial time, so the same holds for the output.
\end{proof}

\section{Discussion}
We supplement Gaussian states with a simple approximation based on the normalized Frobenius norm $\sigma$ to obtain an $\Omega(1/\sqrt{n})$-approximation for homogeneous quartic fermionic Hamiltonians. Is such an approximation possible using only Gaussian states? If so, is this possible in polynomial time?

While our results do not rely on SDP hierarchies, they make critical use of moment and coefficient matrices. Are there broader applications of our moment-matrix bound \eqref{eq:moment-bound-OV}? Can our results be improved by using SDP hierarchies and sums-of-squares bounds?

\section*{AI methodology} GPT-6 Astra was used extensively for this work. I proposed research directions for approximation algorithms for homogeneous quartic fermionic Hamiltonians, and GPT developed additive bounds of the form \eqref{eq:main-additive} and presented an $\Omega(1/n^{3/4})$-approximation algorithm. Through further discussion, we improved this to $\Omega(1/\sqrt{n})$, retaining the same type of additive bound while the proof techniques changed quite a bit. The initial proofs were difficult for me to follow, and I guided and contributed to their reformulation and simplification, often favoring clearer and broadly accessible arguments over sharper constants. For example, I suggested a simpler optimization formulation of the variational argument underlying our main technical lemma, \Cref{lem:local-overlap}, leading to a short, elementary proof. The results and proofs reached their present form through sustained exchanges over approximately a month. This account was reviewed with GPT against the available conversation history.    

\section*{Acknowledgements}
I thank Maarten Stroeks and Yaroslav Herasymenko for conversations that reminded me to revisit using LLMs for this work. 

This work is supported by the U.S.\ Department of Energy (DOE), Office of Science, National Quantum Information Science Research Centers, Quantum Systems Accelerator (QSA). Additional support is acknowledged from DOE, Office of Science, Accelerated Research in Quantum Computing, Fundamental Algorithmic Research toward Quantum Utility (FAR-Qu). 

This article has been authored by an employee of National Technology \& Engineering Solutions of Sandia, LLC under Contract No.\ DE-NA0003525 with the U.S. Department of Energy (DOE). The employee owns all right, title and interest in and to the article and is solely responsible for its contents. The United States Government retains and the publisher, by accepting the article for publication, acknowledges that the United States Government retains a non-exclusive, paid-up, irrevocable, world-wide license to publish or reproduce the published form of this article or allow others to do so, for United States Government purposes. The DOE will provide public access to these results of federally sponsored research in accordance with the DOE Public Access Plan \url{https://www.energy.gov/downloads/doe-public-access-plan}.

\printbibliography

@article{BGKT,
  author = {Sergey Bravyi and David Gosset and Robert K{\"o}nig and Kristan Temme},
  title = {Approximation algorithms for quantum many-body problems},
  journal = {Journal of Mathematical Physics},
  volume = {60},
  number = {3},
  pages = {032203},
  year = {2019},
  doi = {10.1063/1.5085428},
  url = {https://arxiv.org/abs/1808.01734v1},
  note = {Theorem and equation numbering follows arXiv:1808.01734v1}
}

@article{Bravyi,
  author = {Sergey Bravyi},
  title = {Lagrangian representation for fermionic linear optics},
  journal = {Quantum Information and Computation},
  volume = {5},
  number = {3},
  pages = {216--238},
  year = {2005},
  doi = {10.26421/QIC5.3-3},
  url = {https://arxiv.org/abs/quant-ph/0404180v2}
}

@inproceedings{HO,
  author = {Matthew B. Hastings and Ryan O'Donnell},
  title = {Optimizing Strongly Interacting Fermionic {Hamiltonians}},
  booktitle = {Proceedings of the 54th Annual ACM SIGACT Symposium on Theory of Computing (STOC)},
  pages = {776--789},
  year = {2022},
  publisher = {ACM},
  doi = {10.1145/3519935.3519960},
  url = {https://arxiv.org/abs/2110.10701v4},
  note = {Full version arXiv:2110.10701v4; theorem numbering follows that version}
}

@article{DKKA,
  author = {Matthew Ding and Robbie King and Bobak T. Kiani and Eric R. Anschuetz},
  title = {Optimizing Sparse {SYK}},
  journal = {Quantum},
  volume = {10},
  pages = {2029},
  year = {2026},
  doi = {10.22331/q-2026-03-16-2029},
  url = {https://arxiv.org/abs/2506.09037v3}
}

@article{HSHT,
  author = {Yaroslav Herasymenko and Maarten Stroeks and Jonas Helsen and Barbara Terhal},
  title = {Optimizing sparse fermionic {Hamiltonians}},
  journal = {Quantum},
  volume = {7},
  pages = {1081},
  year = {2023},
  doi = {10.22331/q-2023-08-10-1081},
  url = {https://arxiv.org/abs/2211.16518}
}

@inproceedings{HPT,
  author = {Daniel Hothem and Ojas Parekh and Kevin Thompson},
  title = {Improved Approximations for Extremal Eigenvalues of Sparse {Hamiltonians}},
  booktitle = {18th Conference on the Theory of Quantum Computation, Communication and Cryptography (TQC 2023)},
  series = {Leibniz International Proceedings in Informatics (LIPIcs)},
  volume = {266},
  pages = {6:1--6:10},
  year = {2023},
  publisher = {Schloss Dagstuhl -- Leibniz-Zentrum f{\"u}r Informatik},
  doi = {10.4230/LIPIcs.TQC.2023.6},
  url = {https://arxiv.org/abs/2301.04627}
}

@misc{STH,
  author = {Maarten Stroeks and Barbara M. Terhal and Yaroslav Herasymenko},
  title = {Optimizing fermionic {Hamiltonians} with classical interactions},
  year = {2025},
  note = {arXiv:2510.02122v2},
  url = {https://arxiv.org/abs/2510.02122v2}
}

@inproceedings{So,
  author = {Anthony Man-Cho So},
  title = {Improved Approximation Bound for Quadratic Optimization Problems with Orthogonality Constraints},
  booktitle = {Proceedings of the Twentieth Annual ACM-SIAM Symposium on Discrete Algorithms (SODA)},
  pages = {1201--1209},
  year = {2009},
  publisher = {SIAM},
  doi = {10.1137/1.9781611973068.130}
}

@misc{BassoChenDalzell2024,
  author = {Joao Basso and Chi-Fang Chen and Alexander M. Dalzell},
  title = {Optimizing random local {Hamiltonians} by dissipation},
  year = {2024},
  eprint = {2411.02578},
  archivePrefix = {arXiv},
  primaryClass = {quant-ph},
  doi = {10.48550/arXiv.2411.02578},
  url = {https://arxiv.org/abs/2411.02578}
}

@article{McArdle2020,
  author = {Sam McArdle and Suguru Endo and Al{\'a}n Aspuru-Guzik and Simon C. Benjamin and Xiao Yuan},
  title = {Quantum computational chemistry},
  journal = {Reviews of Modern Physics},
  volume = {92},
  number = {1},
  pages = {015003},
  year = {2020},
  doi = {10.1103/RevModPhys.92.015003},
  url = {https://arxiv.org/abs/1808.10402}
}

@article{Arovas2022,
  author = {Daniel P. Arovas and Erez Berg and Steven A. Kivelson and Srinivas Raghu},
  title = {The {Hubbard} Model},
  journal = {Annual Review of Condensed Matter Physics},
  volume = {13},
  pages = {239--274},
  year = {2022},
  doi = {10.1146/annurev-conmatphys-031620-102024},
  url = {https://arxiv.org/abs/2103.12097}
}

@article{LCV2007,
  author = {Yi-Kai Liu and Matthias Christandl and Frank Verstraete},
  title = {Quantum Computational Complexity of the {$N$}-Representability Problem: {QMA} Complete},
  journal = {Physical Review Letters},
  volume = {98},
  number = {11},
  pages = {110503},
  year = {2007},
  doi = {10.1103/PhysRevLett.98.110503},
  url = {https://arxiv.org/abs/quant-ph/0609125}
}

@article{PNA2010,
  author = {Stefano Pironio and Miguel Navascu{\'e}s and Antonio Ac{\'i}n},
  title = {Convergent relaxations of polynomial optimization problems with non-commuting variables},
  journal = {SIAM Journal on Optimization},
  volume = {20},
  number = {5},
  pages = {2157--2180},
  year = {2010},
  doi = {10.1137/090760155},
  url = {https://arxiv.org/abs/0903.4368v2}
}

@misc{PT2022,
  author = {Ojas Parekh and Kevin Thompson},
  title = {An Optimal Product-State Approximation for {2-Local Quantum Hamiltonians} with Positive Terms},
  year = {2022},
  eprint = {2206.08342},
  archivePrefix = {arXiv},
  primaryClass = {quant-ph},
  url = {https://arxiv.org/abs/2206.08342v2},
  note = {Revised March 2026; section numbering follows arXiv:2206.08342v2}
}

@misc{Nielsen2005,
  author = {Michael A. Nielsen},
  title = {The {Fermionic} canonical commutation relations and the {Jordan--Wigner} transform},
  date = {2005-07-29},
  url = {https://futureofmatter.com/assets/fermions_and_jordan_wigner.pdf},
  note = {Expository notes}
}

@article{HTS,
  author = {Arijit Haldar and Omid Tavakol and Thomas Scaffidi},
  title = {Variational wave functions for {Sachdev--Ye--Kitaev} models},
  journal = {Physical Review Research},
  volume = {3},
  number = {2},
  pages = {023020},
  year = {2021},
  doi = {10.1103/PhysRevResearch.3.023020},
  url = {https://arxiv.org/abs/2009.03924v2}
}
\end{document}